\documentclass[10pt,reprint, pra, aps]{revtex4-1}
\usepackage[utf8]{inputenc}
\usepackage[english]{babel}
\usepackage{amsmath}
\usepackage{amsfonts}
\usepackage{amssymb}
\usepackage{makeidx}
\usepackage{graphicx}
\usepackage{amsthm}
\newtheorem{prop}{Proposition}
\newtheorem{coro}{Corollary}
\theoremstyle{remark}
\newtheorem{exm}{Example}
\theoremstyle{definition}
\newtheorem{defin}{Definition}
\def\states{\mathfrak{S}}
\def\Pe{\mathcal{P}}
\def\Ha{\mathcal{H}}
\def\Ce{\mathcal{C}}
\def\conv{\mathit{conv}}

\def\aff{\mathit{aff}}
\def\degcom{\mathit{DegCom}}
\def\<{\langle}
\def\>{\rangle}
\DeclareMathOperator{\Tr}{Tr}
\usepackage{dsfont}
\def\I{\mathds{1}}

\begin{document}

\title{Conditions on the existence of maximally incompatible two-outcome measurements in General Probabilistic Theory}

\author{Anna Jen\v{c}ov\'{a}}
\affiliation{Mathematical Institute, Slovak Academy of Sciences, \v Stef\' anikova 49, Bratislava, Slovakia}
\author{Martin Pl\'{a}vala}
\affiliation{Mathematical Institute, Slovak Academy of Sciences, \v Stef\' anikova 49, Bratislava, Slovakia}

\begin{abstract}
We formulate the necessary and sufficient conditions for the existence of a pair of maximally incompatible two-outcome measurements in a finite dimensional General Probabilistic Theory. The conditions are on the geometry of the state space; they require existence of two pairs of parallel exposed faces with additional condition on their intersections. We introduce the notion of discrimination measurement and show that the conditions for a pair of two-outcome measurements to be maximally incompatible are equivalent to requiring that a (potential, yet non-existing) joint measurement of the maximally incompatible measurements would have to discriminate affinely dependent points. We present several examples to demonstrate our results.
\end{abstract}

\maketitle

\section{Introduction}
General Probabilistic Theories have recently gained a lot of attention. It was identified that several non-classical effects that we know from Quantum Mechanics, such as steering and Bell nonlocality \cite{WisemanJonesDoherty-nonlocal}, can be found in most General Probabilistic Theories. Moreover, it was shown that  one can violate even the bounds we know from Quantum Mechanics. In finite dimensional Quantum Mechanics the minimal degree of compatibility of measurements is bounded below by a dimension-dependent constant \cite{HeinosaariSchultzToigoZiman-maxInc}, while a General Probabilistic Theory may admit pairs of maximally incompatible two-outcome measurements \cite{BuschHeinosaariSchultzStevens-compatibility}, i.e.  two-outcome measurements such that their degree of compatibility attains the minimal value $\frac{1}{2}$.

In the present article we show necessary and sufficient conditions for a pair of maximally incompatible measurements to exist in a given General Probabilistic Theory. The conditions restrain the possible geometry of the state space. We also introduce the notion of discrimination two-outcome measurement and show how the concept of discrimination measurements is connected to maximally incompatible measurements. Our results are demonstrated on some examples. In particular, it is shown that maximally incompatible measurements exist for quantum channels. A somewhat different notion of compatibility of measurements on quantum channels and combs has been recently researched in \cite{SedlakReitznerChiribellaZiman-compatibility}, where similar results were found.

The article is organized as follows: in Sec. \ref{sec:preliminary} we provide a quick review of General Probabilistic Theory and of the notation we will use, in Sec. \ref{sec:meas} we introduce the two-outcome measurements, in Sec. \ref{sec:degcom} we introduce the degree of compatibility and the linear program for compatibility of two-outcome measurements. In Sec. \ref{sec:maxInc} we formulate and prove the necessary and sufficient conditions for maximally incompatible two-outcome measurements to exist. In Sec. \ref{sec:disc} we introduce the concept of discrimination measurement and we show that two-outcome measurements are maximally incompatible if and only if their joint measurement would have to discriminate affinely dependent points, which is impossible.

\section{Structure of General Probabilistic Theory} \label{sec:preliminary}
General Probabilistic Theories form a general framework that provides a unified description of all physical systems known today. We will present the standard definition of a finite dimensional General Probabilistic Theory in a quick review just to settle the notation.

The central notion is that of a state space, that is a compact convex  subset $K\subset \mathbb{R}^n$, representing the set of states of some system. The convex combinations are  interpreted operationally, see e.g. \cite[Part 2]{HeinosaariZiman-MLQT}.

Let $A(K)$ denote the ordered linear space of affine functions $f:K \to \mathbb{R}$. The order on $A(K)$ is introduced in a natural way; let $f, g \in A(K)$ then $f \geq g$ if and only if $f(x) \geq g(x)$ for all $x \in K$. 
 Let $A(K)^+$ be the positive cone, that is the generating, pointed and convex cone of positive affine functions on $K$. We denote the constant functions by the value they attain, i.e. $1(x)=1$ for all $x \in K$. Let $E(K) = \{ f \in A(K): 1 \geq f \geq 0 \}$ denote the set of effects on $K$.

Let $A(K)^*$ be the dual to $A(K)$ and let $\< \psi, f \>$ denote the value of the functional $\psi\in A(K)^*$ on $f \in A(K)$. Using the cone $A(K)^+$ we define the dual order on $A(K)^*$ as follows: let $\psi_1, \psi_2 \in A(K)^*$, then $\psi_1 \geq \psi_2$ if and only if $\< \psi_1, f \> \geq \< \psi_2, f \>$ for every $f \in A(K)^+$. The dual positive cone is  $A(K)^{*+} = \{ \psi \in A(K)^*: \psi \geq 0 \}$ where $0$ denotes the zero functional, $\<0, f \> = 0$ for all $f \in A(K)$.

Let $x \in K$, then $\phi_x$ will denote the positive and normed functional such that $\< \phi_{x}, f \> = f(x)$. It can be seen that for every functional $\psi \in A(K)^{*+}$ such that $\< \psi, 1 \> = 1$ there is some $x \in K$ such that $\psi = \phi_x$, see \cite[Theorem 4.3]{AsimowEllis}. This implies that the set $\states_K = \{ \phi_x: x \in K \}$ is a base of the cone $A(K)^{*+}$, i.e. for every $\psi \in A(K)^{*+}$, $\psi \neq 0$ there is a unique $x \in K$ and unique $\alpha \in \mathbb{R}$, $\alpha > 0$ such that $\psi = \alpha \phi_x$.

For any $X \subset \mathbb{R}^n$, $\conv(X)$  will denote the convex hull of $X$ and  $\aff(X)$ the affine hull of $X$.

\section{Measurements in General Probabilistic Theory} \label{sec:meas}
Let $K\subset \mathbb{R}^n$ be a state space. A measurement on $K$ is an affine map $m: K \to \Pe(\Omega)$, where $\Omega$ is the sample space, that is a measurable space representing the set of all possible measurement outcomes, and $\Pe(\Omega)$ is the set of all probability measures on $\Omega$.

We will be mostly interested in two-outcome measurements, i.e. measurements with the sample space $\Omega = \{ \omega_1, \omega_2 \}$. Let $\mu \in \Pe ( \Omega )$, then $\mu = \lambda \delta_1 + (1-\lambda) \delta_2$ for some $\lambda \in [0, 1] \subset \mathbb{R}$, where $\delta_1=\delta_{\omega_1}$, $\delta_2=\delta_{\omega_2} $ are the Dirac measures. This shows that the general form of two-outcome measurement $m_f$ is
\begin{equation*}
m_f = f \delta_1 + (1-f) \delta_2
\end{equation*}
for some $f \in E(K)$. Strictly speaking, this should be written as $m_f = f \otimes \delta_1 + (1-f) \otimes \delta_2$, since any map $m_f: K \to \Pe(\Omega)$ can be identified with a point of $A(K)^+ \otimes \Pe(\Omega)$, see e.g. \cite{Ryan-tensProd}. The interpretation is that a point $x \in K$ is mapped to the probability measure $m_f(x) = f(x) \delta_1 + (1-f(x)) \delta_2$, i.e. $f(x)$ corresponds to the probability of measuring the outcome $\omega_1$.

Let $f, g \in E(K)$ and let $m_f$, $m_g$ be the corresponding two-outcome measurements. We will keep this notation throughout the paper. The two-outcome measurements $m_f$, $m_g$ are compatible if and only if there exists a function $p \in E(K)^+$ such that
\begin{align}
f &\geq p, \label{eq:meas-cond-1} \\
g &\geq p, \label{eq:meas-cond-2} \\
1 + p &\geq f + g, \label{eq:meas-cond-3}
\end{align}
see \cite{Plavala-simplex} for a derivation of these conditions from the standard conditions that can be found e.g. in \cite[Chapter 2]{Holevo-QT}.

\begin{prop} \label{prop:meas-postProc}
$m_f$, $m_g$ are compatible if and only if $m_{(1-f)}$, $m_g$ are compatible.
\end{prop}
\begin{proof} Assume that $m_f$, $m_g$ are compatible and let $p\in E(K)$ satisfy \eqref{eq:meas-cond-1} - \eqref{eq:meas-cond-3}.
Let $p' = g - p$, then Eq. \eqref{eq:meas-cond-2} implies $p' \geq 0$, Eq. \eqref{eq:meas-cond-3} implies $1-f \geq p'$,  $p \geq 0$ implies $g \geq p'$ and \eqref{eq:meas-cond-1} implies $1 + p' \geq (1-f) + g$.

Since $1 - (1-f) = f$ it is clear that the compatibility of $m_{(1-f)}$, $m_g$ implies compatibility of $m_f$, $m_g$ in the same manner.
\end{proof}

\section{Degree of compatibility} \label{sec:degcom}
\begin{defin}
A coin-toss measurement on $K$ is a constant map 
\begin{equation*}
\tau (x) = \mu \in \Pe(\Omega),\qquad x\in K.
\end{equation*}
\end{defin}

It is straightforward that a coin-toss measurement is compatible with any other measurement.
In the following we state the usual definition of the degree of compatibility.
\begin{defin}
Let $m_f, m_g$ be two-outcome measurement on $K$ with sample space $\Omega$. The  degree of compatibility of  $m_f, m_g$ is defined as
\begin{align*}
\degcom (m_f, m_g) =& \sup_{\substack{0 \leq \lambda \leq 1 \\ \tau_1, \tau_2}} \{ \lambda : \lambda m_f + (1-\lambda) \tau_1, \\ & \lambda m_g + (1-\lambda) \tau_2 \; \text{are compatible} \}
\end{align*}
where the supremum is taken over all coin-toss measurements $\tau_1, \tau_2$.
\end{defin}
It is easy to see that for every two measurements $m_f$, $m_g$ we always have $\degcom(m_f, m_g) \geq \frac{1}{2}$, see e.g. \cite{HeinosaariMiyaderaZiman-compatibility}. The following immediately follows from Prop. \ref{prop:meas-postProc}.

\begin{coro} \label{coro:degcom-postProc}
Let $m_f$, $m_g$ be two-outcome measurements, then
\begin{align*}
\degcom(m_f, m_g) &= \degcom(m_{(1-f)}, m_g) \\
&= \degcom(m_f, m_{(1-g)}) \\
&= \degcom(m_{(1-f)}, m_{(1-g)}).
\end{align*}
\end{coro}

\begin{defin}
We will say that two measurements are maximally incompatible if $\degcom(m_f, m_g) = \frac{1}{2}$.
\end{defin}
It is known that such measurements exist for some state spaces \cite{BuschHeinosaariSchultzStevens-compatibility}, but they do not exist in finite dimensional Quantum Mechanics \cite{HeinosaariSchultzToigoZiman-maxInc}.

Let $\tau = \frac{1}{2}(\delta_1 + \delta_2)$, then we define
\begin{align*}
\degcom_{\frac{1}{2}} (m_f, m_g) =& \sup_{0 \leq \lambda \leq 1} \{ \lambda : \lambda m_f + (1-\lambda) \tau, \\ & \lambda m_g + (1-\lambda) \tau \; \text{are compatible} \}
\end{align*}
as the degree of compatibility provided only by mixing the measurements $m_f$, $m_g$ with the coin-toss measurement $\tau$. Clearly we have
\begin{equation*}
\degcom_{\frac{1}{2}} (m_1, m_2) \leq \degcom (m_f, m_g)
\end{equation*}
so $\degcom_{\frac{1}{2}} (m_f, m_g) = 1$ implies $\degcom (m_f, m_g) = 1$ and $\degcom (m_f, m_g) = \frac{1}{2}$ implies $\degcom_{\frac{1}{2}} (m_f, m_g) = \frac{1}{2}$.

In \cite{Plavala-simplex} it was shown that the dual linear program for the compatibility of the measurements $m_f$, $m_g$ is of the form
\begin{align*}
&\sup \big( a_3 (f(z_3) + g(z_3) - 1) - a_1 f(z_1) - a_2 g(z_2) \big) \\
&a_1 + a_2 \leq 2 \\
&a_3 \phi_{z_3} \leq a_1 \phi_{z_1} + a_2 \phi_{z_2}
\end{align*}
where $z_1, z_2, z_3 \in K$ and $a_1, a_2, a_3$ are non-negative numbers. Let $\beta$ denote the supremum, then we have
\begin{equation*}
\beta = \dfrac{1 - \degcom_{\frac{1}{2}} (m_f, m_g)}{\degcom_{\frac{1}{2}} (m_f, m_g)},
\end{equation*}
i.e. $\beta = 0$ if the measurements are compatible and $\beta > 0$ implies that the measurements are incompatible.

Assume that the measurements $m_f$, $m_g$ are incompatible. Then we have $\beta > 0$, which implies $a_1 + a_2 > 0$. We will show  that if the supremum is reached, we must have $a_1 + a_2 = 2$. Assume that the supremum is reached for some $a_1, a_2, a_3$ and $z_1, z_2, z_3$ such that $a_1 + a_2 < 2$ and define
\begin{align*}
a_1' &= \dfrac{2}{a_1 + a_2} a_1, \\
a_2' &= \dfrac{2}{a_1 + a_2} a_2, \\
a_3' &= \dfrac{2}{a_1 + a_2} a_3. \\
\end{align*}
It is straightforward to see that $a_3' \phi_{z_3} \leq a_1' \phi_{z_1} + a_2' \phi_{z_2}$, moreover
\begin{align*}
\beta < a_3' (f(z_3) + g(z_3) - 1) - a_1' f(z_1) - a_2' g(z_2)
\end{align*}
which is a contradiction.

It follows that in the case when the measurements $m_f$, $m_g$ are incompatible we can write the linear program as
\begin{align}
&\sup 2 \big( \eta (f(z_3) + g(z_3) - 1) - \nu f(z_1) - (1-\nu) g(z_2) \big) \nonumber \\
&\eta \phi_{z_3} \leq \nu \phi_{z_1} + (1-\nu) \phi_{z_2} \label{eq:degcom-linProg-final}
\end{align}
where we have set $a_1 + a_2 = 2$ and used a substitution $2 \nu = a_1$, $2 \eta = a_3$. Also note that $\eta \phi_{z_3} \leq \nu \phi_{z_1} + (1-\nu) \phi_{z_2}$ implies that there exists $z_4 \in K$ such that
\begin{equation*}
\nu z_1 + (1-\nu) z_2 = \eta z_3 + (1-\eta) z_4.
\end{equation*}

\section{Maximally incompatible two-outcome measurements} \label{sec:maxInc}
In this section, we wish to find conditions on the state space, under which $\degcom(m_f, m_g) = \frac{1}{2}$  for a pair of two-outcome measurements $m_f, m_g$. A sufficient condition was proved in \cite{BuschHeinosaariSchultzStevens-compatibility}: a pair of maximally incompatible two-outcome measurements exists if the state space $K$ is a square, or more generally, there are two pairs of parallel hyperplanes tangent to $K$, such that the corresponding exposed faces contain the edges of a square. Here a square is defined as the convex hull of four points $x_1$, $x_2$, $x_3$, $x_4$ satisfying $x_1+x_2=x_3+x_4$. Besides the square, such state spaces include the cube, pyramid, double pyramid, cylinder etc. We will show that this condition is also necessary, so that it characterizes state spaces admitting a pair of maximally incompatible two-outcome measurements.

Let us fix a pair of effects $f,g\in E(K)$. The following notation  will be used throughout. 

\begin{align*}
F_0 &= \{ z \in K : f(z) = 0 \}, \\
F_1 &= \{ z \in K : f(z) = 1 \}, \\
G_0 &= \{ z \in K : g(z) = 0 \}, \\
G_1 &= \{ z \in K : g(z) = 1 \}.
\end{align*}

We begin by rephrasing the above sufficient condition. For completeness, we add a proof along the lines of \cite{BuschHeinosaariSchultzStevens-compatibility}.

\begin{prop} \label{prop:maxInc-necessary}
Assume there are some points  $x_{00} \in F_0 \cap G_0$, $x_{10} \in F_1 \cap G_0$, $x_{01} \in F_0 \cap G_1$, $x_{11} \in F_1 \cap G_1$ such that
\begin{equation*}
\dfrac{1}{2} ( x_{00} + x_{11} ) = \dfrac{1}{2} ( x_{10} + x_{01} ).
\end{equation*}
Then $\degcom(m_f, m_g) = \frac{1}{2}$.
\end{prop}
\begin{proof}
 Let $p$ be any positive affine function on $K$, then we have
\begin{equation*}
p(x_{11}) + p(x_{00}) = p(x_{10}) + p(x_{01})
\end{equation*}
and
\begin{equation*}
p(x_{11}) \leq p(x_{10}) + p(x_{01})
\end{equation*}
follows. Let $\tau_1 = \mu_1 \delta_{\omega_1} + (1-\mu_1) \delta_{\omega_2}$ and $\tau_2 = \mu_2 \delta_{\omega_1} + (1-\mu_2) \delta_{\omega_2}$ be coin-toss measurements, then the conditions \eqref{eq:meas-cond-1} - \eqref{eq:meas-cond-3} for  $\lambda m_f + (1-\lambda) \tau_1$ and $\lambda m_g + (1-\lambda) \tau_2$ take the form
\begin{align*}
\lambda f + (1-\lambda) \mu_1 &\geq p, \\
\lambda g + (1-\lambda) \mu_2 &\geq p, \\
1 + p &\geq \lambda (f + g) + (1-\lambda)(\mu_1 + \mu_2).
\end{align*}
Expressing some of these conditions at the points $x_{10}, x_{01}, x_{11}$ we get
\begin{align}
1 + p(x_{11}) &\geq 2 \lambda + (1-\lambda)(\mu_1 + \mu_2), \label{eq:max-Inc-necessary-x11} \\
(1-\lambda) \mu_1 &\geq p(x_{01}), \label{eq:max-Inc-necessary-x01} \\
(1-\lambda) \mu_2 &\geq p(x_{10}). \label{eq:max-Inc-necessary-x10}
\end{align}
From \eqref{eq:max-Inc-necessary-x11} we obtain
\begin{equation*}
2 \lambda \leq 1 + p(x_{11}) - (1-\lambda)(\mu_1 + \mu_2)
\end{equation*}
and since from \eqref{eq:max-Inc-necessary-x01} and \eqref{eq:max-Inc-necessary-x10} we have
\begin{equation*}
p(x_{11}) \leq p(x_{10}) + p(x_{01}) \leq (1-\lambda)(\mu_1 + \mu_2)
\end{equation*}
it follows that $\lambda \leq \frac{1}{2}$.
\end{proof}

At this point we can demonstrate that  maximally incompatible two-outcome 
measurements exist for the set of quantum channels. We will use the standard and well-know definitions that may be found in \cite{HeinosaariZiman-MLQT}.
\begin{exm} \label{ex:maxInc-channels}
Let $\Ha$ denote a finite dimensional complex Hilbert space, let $B_h(\Ha)$ denote the set of self-adjoint operators on $\Ha$ and let $\I$ denote the identity operator. Let $A \in B_h(\Ha)$, then we write $A \geq 0$ if $A$ is positive semi-definite. Let $\Ha \otimes \Ha$ denote the tensor product of $\Ha$ with itself and let $\Tr_1$ denote the partial trace. Let
\begin{equation*}
\Ce_\Ha = \{ C \in B_h(\Ha \otimes \Ha): \Tr_1 (C) = \I, C \geq 0 \}
\end{equation*}
be the set of Choi matrices of all channels $B_h(\Ha) \to B_h(\Ha)$. This is clearly a finite dimensional state space.
The effects  $f\in E(\Ce_\Ha)$ have the form $f(C)=\Tr CM$, $C\in \Ce_\Ha$, where  $M\in B_h(\Ha\otimes\Ha)$ is 
such that 
\[
0\le M\le \I\otimes \sigma,
\]
for some density operator $\sigma$  on $\Ha$, \cite{Jencova-genChannels,HeinosaariMiyaderaZiman-compatibility}, so that effects are given by two-outcome PPOVMs defined in \cite{Ziman-ppovm}. 

Let $\dim(\Ha) = 2$, let $|0\>, |1\>$ be an orthonormal basis of $\Ha$ and let 
$M,N \in B_h(\Ha \otimes \Ha)$ be given as
\begin{align*}
M&= |0\>\<0| \otimes |0\>\<0|, \\
N &= |0\>\<0| \otimes |1\>\<1|. \\
\end{align*}
Then  $0\le M\le  \I \otimes |0\>\<0|$ and $0\le N\le  \I \otimes |1\>\<1|$, so that $f(C)=\Tr CM$ and $g(C)=\Tr CN$ define effects on $\Ce_\Ha$. Let
\begin{align*}
C_{00} &= |1\>\<1| \otimes \I\\
C_{10} &= |0\>\<0| \otimes |0\>\<0| + |1\>\<1| \otimes |1\>\<1|, \\
C_{01} &= |0\>\<0| \otimes |1\>\<1| + |1\>\<1| \otimes |0\>\<0|, \\
C_{11} &= |0\>\<0| \otimes \I.
\end{align*}
It is easy to check that $C_{00}, C_{10}, C_{01}, C_{11} \in \Ce_\Ha$. Moreover
\begin{equation*}
C_{00} + C_{11} = \I \otimes \I = C_{10} + C_{01}
\end{equation*}
and 
\begin{align*}
\Tr(C_{00} M) = \Tr(C_{00} N) &= 0, \\
\Tr(C_{10} M) =1,\  \Tr(C_{10} N) &= 0, \\
\Tr(C_{01} M) =0,\  \Tr(C_{01} N) &= 1, \\
\Tr(C_{11}M) = \Tr(C_{11} N) &= 1.
\end{align*}
In conclusion,  $C_{00}$, $C_{10}$, $C_{01}$, $C_{11}$ satisfies the properties in Prop. \ref{prop:maxInc-necessary}, so that 
the two-outcome measurements $m_f$ and $m_g$ are maximally incompatible.

Analogical fact was also observed in \cite{SedlakReitznerChiribellaZiman-compatibility, HeinosaariMiyaderaZiman-compatibility} in different circumstances.
\end{exm}

We proceed to prove some necessary conditions.
\begin{prop} \label{prop:maxInc-notEmpty}
$\degcom(m_f, m_g) = \frac{1}{2}$ only if
\begin{align*}
F_0 \cap G_0 &\neq \emptyset, \\
F_0 \cap G_1 &\neq \emptyset, \\
F_1 \cap G_0 &\neq \emptyset, \\
F_1 \cap G_1 &\neq \emptyset.
\end{align*}
\end{prop}
\begin{proof}
Let $F_1 \cap G_1 = \emptyset$, then $f+g < 2$. Let $\tau = \delta_{\omega_2}$ and consider the measurements $\lambda m_f + (1-\lambda)\tau= m_{\lambda f}$ and $\lambda m_g + (1-\lambda) \tau= m_{\lambda g}$, $\lambda \in [0, 1]$.  Since $f+g < 2$, we can choose $\lambda > \frac{1}{2}$ such that $1 \geq \lambda(f+g)$, so that 
$p = 0$ satisfies  equations Eq. \eqref{eq:meas-cond-1} - \eqref{eq:meas-cond-3} for $m_{\lambda f}$, $m_{\lambda g}$.

The result for the other sets follows by using the Corollary \ref{coro:degcom-postProc}.
\end{proof}

The conditions given by the Prop. \ref{prop:maxInc-notEmpty} are not sufficient as we will demonstrate in the following example.
\begin{exm}
Let $K$ be a simplex with the vertices $x_1, x_2, x_3, x_4$ and let $b_1, b_2, b_3, b_4$ denote positive affine functions such that
\begin{equation*}
b_i (x_j) = \delta_{ij}.
\end{equation*}
Such functions exist because $K$ is a simplex. Let
\begin{align*}
f &= b_1 + b_2, \\
g &= b_1 + b_3,
\end{align*}
then we have
\begin{align*}
F_1 \cap G_1 &= \{ x_1 \}, \\
F_1 \cap G_0 &= \{ x_2 \}, \\
F_0 \cap G_1 &= \{ x_3 \}, \\
F_0 \cap G_0 &= \{ x_4 \},
\end{align*}
but clearly the measurements $m_f$ and $m_g$ must be compatible as $K$ is a simplex. Matter of fact, the Eq. \eqref{eq:meas-cond-1} - \eqref{eq:meas-cond-3} are satisfied with $p = b_1$.
\end{exm}

\begin{prop} \label{prop:maxInc-sufficient}
 $\degcom(m_f, m_g) = \frac{1}{2}$ if and only if  there exist points $x_{00}, x_{01}, x_{10}, x_{11}$ such that $x_{00} \in F_0 \cap G_0$, $x_{10} \in F_1 \cap G_0$, $x_{01} \in F_0 \cap G_1$, $x_{11} \in F_1 \cap G_1$ and
\begin{equation*}
\dfrac{1}{2} ( x_{00} + x_{11} ) = \dfrac{1}{2} ( x_{10} + x_{01} ).
\end{equation*}
\end{prop}
\begin{proof} The 'if' part is proved in Prop. \ref{prop:maxInc-necessary}. Conversely,
if $\degcom(m_f, m_g) = \frac{1}{2}$ then according to the results in Section \ref{sec:degcom}, the supremum in  \eqref{eq:degcom-linProg-final}  must be equal to 1, so that we must have 
\begin{equation}
\eta( f(z_3) + g(z_3) - 1) - \nu f(z_1) - (1-\nu) g(z_2) = \dfrac{1}{2} \label{eq:maxInc-functions}
\end{equation}
for some $\eta, \nu \in [0, 1]$ and $z_1, z_2, z_3 \in K$, such that
\begin{equation*}
\nu \phi_{z_1} + (1-\nu) \phi_{z_2} \geq \eta \phi_{z_3}.
\end{equation*}
It follows that
\begin{align}
\nu \phi_{z_1} \geq \eta \phi_{z_3} - (1-\nu) \phi_{z_2}, \label{eq:maxInc-pointIneq-mod-1} \\
(1-\nu) \phi_{z_2} \geq \eta \phi_{z_3} - \nu \phi_{z_1}. \label{eq:maxInc-pointIneq-mod-2}
\end{align}
Rewriting  Eq. \eqref{eq:maxInc-functions} we get
\begin{equation*}
\<\eta \phi_{z_3} - \nu \phi_{z_1},f\> + \<\eta \phi_{z_3} - (1-\nu) \phi_{z_2},g\> - \eta = \dfrac{1}{2}.
\end{equation*}
We clearly have $\<\eta \phi_{z_3} - \nu \phi_{z_1},f\> \leq \eta$, but Eq. \eqref{eq:maxInc-pointIneq-mod-2} implies $\<\eta \phi_{z_3} - \nu \phi_{z_1},f\> \leq 1-\nu$ and thus we must have $\<\eta \phi_{z_3} - \nu \phi_{z_1},f\> \leq \min(\eta, 1-\nu)$. Similarly, 
we get $\<\eta \phi_{z_3} - (1-\nu) \phi_{z_2},g\> \leq \min(\eta, \nu)$ and
\begin{equation*}
\dfrac{1}{2} \leq \min(\eta, \nu) + \min(\eta, 1-\nu) - \eta = \min(\nu, 1-\nu, \eta, 1-\eta)
\end{equation*}
which implies $\nu = \eta = \frac{1}{2}$. Moreover, there must be some $z_4 \in K$ such that
\begin{equation}
\dfrac{1}{2}( z_1 + z_2 ) = \dfrac{1}{2} ( z_3 + z_4 ). \label{eq:maxInc-squarePoints}
\end{equation}
Eq. \eqref{eq:maxInc-functions} now becomes
\begin{equation*}
f(z_3) + g(z_3) - f(z_1) - g(z_2) = 2
\end{equation*}
which implies $f(z_3) = g(z_3) = 1$ and $f(z_1) = g(z_2) = 0$ as $f, g \in E(K)$. From Eq. \eqref{eq:maxInc-squarePoints} we get
\begin{equation*}
f(z_2) = 1 + f(z_4),
\end{equation*}
which implies $f(z_2) = 1$, $f(z_4) = 0$ and
\begin{equation*}
g(z_1) = 1 + g(z_4),
\end{equation*}
which implies $g(z_1) = 1$, $g(z_4) = 0$. Together we get
\begin{align*}
z_3 \in F_1 \cap G_1, \\
z_2 \in F_1 \cap G_0, \\
z_1 \in F_0 \cap G_1, \\
z_4 \in F_0 \cap G_0.
\end{align*}
\end{proof}

In the remainder of this section, we aim to give some geometric interpretation of the condition in Prop. \ref{prop:maxInc-sufficient}.

\begin{coro} \label{coro:maxInc-parallelogram}
Let $S \subset \mathbb{R}^2$ be a state space, then maximally incompatible two-outcome measurements exist on $S$ if and only if $S$ is a parallelogram.
\end{coro}
\begin{proof}
Assume that maximally incompatible measurements exist on $S$, then it is clear that $S$ must contain 4 exposed faces, such that each of them has a nonempty intersection with other two and is disjoint from the third. It follows that the faces must be line segments and that $K$ is a polygon with 4 vertices corresponding to the intersections of the aforementioned exposed faces. The opposite edges of $S$ must be parallel as they are the intersections of $S$ with the hyperplanes $\{x \in \mathbb{R}^2: f(x) = 0 \}$ and $\{x \in \mathbb{R}^2: f(x) = 1 \}$ for some $f \in E(S)$.

Assume that $S$ is a parallelogram, then its vertices $z_1, z_2, z_3, z_4$ must satisfy
\begin{equation*}
\frac{1}{2} ( z_1 + z_2 ) = \frac{1}{2} ( z_3 + z_4 )
\end{equation*}
by definition. It follows that maximally incompatible measurements on $S$ exist by Prop. \ref{prop:maxInc-sufficient}.
\end{proof}

\begin{prop} \label{prop:maxInc-KcapVeqS}
Maximally incompatible two-outcome measurements on $K$ exists only if there is an affine subspace $V \subset \aff(K)$, $\dim(V) = 2$ such that $V \cap K = S$, where $S$ is a parallelogram.
\end{prop}
\begin{proof}
The principal idea is that $V = \aff(S)$ where $S$ is the parallelogram in question.

Lets assume that there exist maximally incompatible measurements $m_f$, $m_g$ on $K$ and let $x_{00}, x_{01}, x_{10}, x_{11}$ be the four points as in Proposition \ref{prop:maxInc-sufficient}. Let $V = \aff(x_{00}, x_{10}, x_{01})$ and let $S = V \cap K$, we will show that $S = \conv(x_{00}, x_{01}, x_{10}, x_{11})$. Let $y \in S$, then we must have
\begin{equation*}
y = \alpha_{10} x_{10} + \alpha_{01} x_{01} + (1 - \alpha_{10} - \alpha_{01}) x_{00}
\end{equation*}
for some $\alpha_{10}, \alpha_{01} \in \mathbb{R}$. Since $f, g \in E(S)$ we must have $\alpha_{10}, \alpha_{01} \in [0,1]$ which implies $(1 - \alpha_{10} - \alpha_{01}) \in [-1, 1]$. If $(1 - \alpha_{10} - \alpha_{01}) \in [0, 1]$ then $y \in \conv(x_{00}, x_{10}, x_{01})$. If $(1 - \alpha_{10} - \alpha_{01}) \in [-1, 0]$, then
\begin{equation*}
y = ( 1 - \alpha_{01}) x_{10} + (1 - \alpha_{10}) x_{01} - (1 - \alpha_{10} - \alpha_{01}) x_{11}
\end{equation*}
and $y \in \conv(x_{10}, x_{01}, x_{11})$. It follows that $S$ is a parallelogram by Corollary \ref{coro:maxInc-parallelogram}.
\end{proof}

We will present an example to show that the  condition in Prop. \ref{prop:maxInc-KcapVeqS} is not sufficient, even if the 
parallelogram $S$ is an exposed face of $K$.

\begin{exm} \label{ex:maxInc-cutOffPyramind}
Let $K \subset \mathbb{R}^3$ defined as
\begin{align*}
K = \conv ( \{ &(0, 0, 0), (2, 0, 0), (0, 2, 0), \\
&(2, 1, 0), (1, 2, 0), (1, 1, 1), \\
&(1, 0, 1), (0, 1, 1), (0, 0, 1) \} ),
\end{align*}
see Fig. \ref{fig:cutOffPyramid}. Let
\begin{equation*}
V = \left\lbrace (x_1, x_2, x_3) \in \mathbb{R}^3 : x_3 = 1 \right\rbrace
\end{equation*}
then $K \cap V = S$ where
\begin{equation*}
S = \conv ( \{  (1, 1, 1), (1, 0, 1), (0, 1, 1), (0, 0, 1) \} )
\end{equation*}
is an exposed face and a square.

To see that there is not a pair of maximally incompatible measurements $m_f$, $m_g$, corresponding to $S$, it is enough to realize that the effects $f$, $g$ would have to reach the values $0$ and $1$ on maximal faces that are not parallel, i.e. we would have to have $\aff(F_0) \cap \aff(F_1) \neq \emptyset$ and $\aff(G_0) \cap \aff(G_1) \neq \emptyset$ which is impossible.
\end{exm}

\begin{figure}
\includegraphics[scale=0.5]{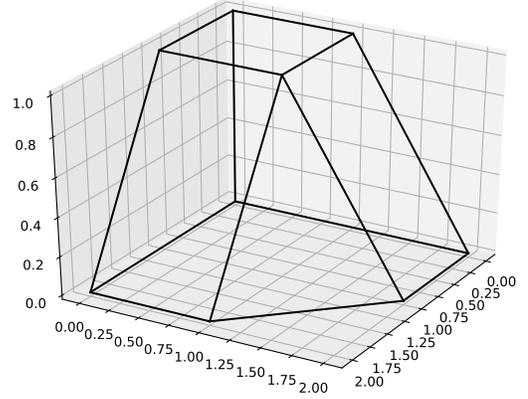}
\caption{The state space $K$ used in Example \ref{ex:maxInc-cutOffPyramind}.} 
\label{fig:cutOffPyramid}
\end{figure}

On the other hand, the examples of double pyramid or a cylinder show that maximally incompatible two-outcome measurements may exist on $K$ even if the parallelogram $S$ described in Prop. \ref{prop:maxInc-KcapVeqS} is not a face of $K$.

\section{Discrimination measurements} \label{sec:disc}
To avoid the problems presented in Example \ref{ex:maxInc-cutOffPyramind} we will introduce a new type of measurement that will allow us to formulate the conditions for existence  of maximally incompatible two-outcome measurements in a clearer way. This will also clarify why the measurements are maximally incompatible.

\begin{defin}
We say that a two-outcome measurement $m_f$ discriminates the sets $E_0, E_1 \subset K$, if it holds that
\begin{align*}
E_0 &\subset \{ x \in K: f(x) = 0 \}, \\
E_1 &\subset \{ x \in K: f(x) = 1 \}.
\end{align*}
We call such measurement a discrimination measurement.
\end{defin}
The idea of the definition is simple: assume that a system is in an unknown state, but we know that it either belongs to $E_0$ or $E_1$. By performing the discrimination measurement $m_f$ we can tell with $100\%$ accuracy whether the state of the system belongs to $E_0$ or $E_1$.

The definition can be generalized to general measurements that can discriminate more than two exposed faces. Most well-known discrimination measurements used in Quantum Mechanics are projective measurements consisting of rank-1 projections that discriminate the states corresponding to the projections.

We are ready to formulate the necessary and sufficient condition for existence of a pair of maximally incompatible two-outcome measurements.
\begin{prop} \label{prop:disc-discJointMeas}
The measurements  $m_f$ and $m_g$ are maximally incompatible  if and only if there is an affine subspace $V \subset \aff(K)$ such that $S=K \cap V$ is a parallelogram and $m_f$ and $m_g$ discriminate the opposite edges of $S$. 
\end{prop}
\begin{proof}
First assume that there is an affine subspace $V$ such that $K \cap V = S$ is a parallelogram whose opposite edges can be discriminated by measurements $m_f$, $m_g$. Denote the vertices of  $S$ as $x_{00}$, $x_{10}$, $x_{01}$, $x_{11}$, then it is clear that  the requirements of Prop. \ref{prop:maxInc-necessary} are satisfied and thus we must have $\degcom(m_f, m_g) = \frac{1}{2}$.

Assume that for some two-outcome measurements $m_f$, $m_g$ we have $\degcom(m_f, m_g) = \frac{1}{2}$. As a result of Prop. \ref{prop:maxInc-KcapVeqS} there must exist an affine subspace $V$ and a parallelogram $S$ such that $K \cap V = S$, moreover it can be seen that we must have
\begin{align*}
E_1 &\subset F_0 \\
E_3 &\subset F_1 \\
E_2 &\subset G_0 \\
E_4 &\subset G_1
\end{align*}
for the edges $E_1$, $E_2$, $E_3$, $E_4$ of $S$. This implies that $m_f$ discriminates $E_1$ and $E_3$, $m_g$ discriminates $E_2$ and $E_4$.
\end{proof}

To give some insight into maximal incompatibility of two-outcome measurements, consider that in general the joint measurement $m$ of $m_f$ and $m_g$ is of the form
\begin{equation*}
m = p \delta_{(1,1)} + (f-p) \delta_{(1,2)} + (g-p) \delta_{(2,1)} + (1+p-f-g) \delta_{(2,2)},
\end{equation*}
for some $p \in E(K)$ as in Eq. \eqref{eq:meas-cond-1} - \eqref{eq:meas-cond-3}, see \cite{Plavala-simplex} for a more detailed calculation. Assume that $m_f$, $m_g$ satisfy the requirements of Prop. \ref{prop:maxInc-necessary}. Inserting  $x_{11} \in F_1 \cap G_1$ into Eq. \eqref{eq:meas-cond-3}, we get
\begin{equation*}
p(x_{11}) \geq 1
\end{equation*}
that together with $p \in E(K)$ implies $p(x_{11}) = 1$. Eq. \eqref{eq:meas-cond-1}, \eqref{eq:meas-cond-2} and the positivity of $p$ imply
\begin{equation*}
p(x_{00}) = p(x_{10}) = p(x_{01}) = 0.
\end{equation*}
Expressing also the functions $(f-p)$, $(g-p)$ and $(1+p-f-g)$ on the points $x_{00}$, $x_{10}$, $x_{01}$, $x_{11}$ we get
\begin{align*}
(f-p)(x_{00}) = (f-p)(x_{01}) = (f-p)(x_{11}) &= 0, \\
(f-p)(x_{10}) &= 1, \\
(g-p)(x_{00}) = (g-p)(x_{10}) = (g-p)(x_{11}) &= 0, \\
(g-p)(x_{01}) &= 1, \\
(1+p-f-g)(x_{00}) = (1+p-f-g)(x_{10}) &= 0, \\
(1+p-f-g)(x_{01}) &= 0, \\
(1+p-f-g)(x_{11}) &= 1.
\end{align*}
This shows that the joint measurement $m$ would have to discriminate the points $x_{00}$, $x_{10}$, $x_{01}$, $x_{11}$, which is generally impossible as they are required to be affinely dependent. We have proved the following:
\begin{coro}
Two-outcome measurements $m_f$, $m_g$ are maximally incompatible if and only if the Eq. \eqref{eq:meas-cond-1} - \eqref{eq:meas-cond-3} imply that the joint measurement would have to discriminate points $x_{00}$, $x_{10}$, $x_{01}$, $x_{11}$ such that
\begin{equation*}
\dfrac{1}{2} (x_{00} + x_{11}) = \dfrac{1}{2} (x_{10} + x_{01}).
\end{equation*}
\end{coro}
This result may also be tied to the way how the joint measurement can be constructed as described also in \cite{HeinosaariMiyaderaZiman-compatibility}; we can toss a fair coin and implement one of the measurements based on the result while ignoring the other. This just roughly shows that maximally incompatible measurements are so incompatible, that the only way how to perform them both requires as much noise as randomly ignoring one of the measurements.

\section{Conclusions}
We have shown necessary and sufficient condition for existence of maximally incompatible two outcome measurements. It turned out that the example of square state space in \cite{BuschHeinosaariSchultzStevens-compatibility} did cover the essence of why and how maximally incompatible measurements come to be in General Probabilistic Theory, but it was also demonstrated by Example \ref{ex:maxInc-cutOffPyramind} that in general more conditions have to be required for more than two dimensional state spaces. 

The geometric interpretation of the minimal degree of compatibility that can be attained on a state space $K$ is an interesting question for future research. It would be also of interest whether the connection between discriminating certain sets and compatibility of measurements could be fruitful from information-theoretic viewpoint. This area of research might also yield some  insight into why there exist maximally incompatible measurements on quantum channels as demonstrated by Example \ref{ex:maxInc-channels}, when they do not exist on quantum states.

\section*{Acknowledgments}
The authors are thankful Teiko Heinosaari as the idea for this calculation emerged during our conversation.

This research was supported by grant VEGA 2/0069/16. MP acknowledges that this research was done during a PhD study at Faculty of Mathematics, Physics and Informatics of the Comenius University in Bratislava.

\bibliography{citations}
\end{document}